\documentclass[english]{smfart}

\theoremstyle{plain}
\newtheorem{thm}{Theorem}
\newtheorem{lem}[thm]{Lemma}
\theoremstyle{remark}
\newtheorem{rem}[thm]{Remark}

\numberwithin{equation}{section}

\newcommand{\R}{\mathbb R}             
\newcommand{\C}{\mathbb C}             

\DeclareMathOperator{\im}{im}

\newcommand{\mshf}[1]{\protect{\frac{1}{2\omega_\mathbf{#1}}}}
\newcommand{\mshi}[1]{\protect{\textup{d}^3\mathbf{#1}}}

\newcommand*{\omsh}[1]{\omega_\mathbf{#1}}

\author{Dorothea Bahns}
\address{Courant Research Centre ``Higher Order Structures in Mathematics'', Mathematisches Institut, Universit\"at G\"ottingen, Bunsenstr. 3-5, D - 37073 G\"ottingen, Germany}
\email{bahns@uni-math.gwdg.de}
\title[Schwinger functions in noncommutative quantum field theory]{Schwinger functions in\\ noncommutative quantum field theory}

\begin{document}
\frontmatter
\begin{abstract}
It is shown that the $n$-point functions of scalar massive free fields on the noncommutative Minkowski space are distributions which are boundary values of analytic functions. Contrary to what one might expect, this construction does not provide a connection to the popular traditional Euclidean approach to noncommutative field theory (unless the time variable is assumed to commute). Instead, one finds Schwinger functions with twistings  involving only momenta that are on the mass-shell. This  explains why renormalization in the traditional Euclidean noncommutative framework crudely differs from renormalization in the Minkowskian regime.
\end{abstract}
\maketitle
\mainmatter
\section{Introduction}
\label{sec:Intro}
A quantum field theoretic model is to a large part determined by 
the choice of a partial differential operator. For physical reasons, this operator has to be  \emph{hyperbolic}, and one of its fundamental solutions, the so-called Feynman propagator, is the building block in any perturbative calculation of  physically relevant quantities. Nonetheless, 
ever since proposed by Symanzik in 1966~\cite{sym} based on ideas of Schwinger, the so-called Euclidean framework has played a very important role. In this framework, the building block is the so-called Schwinger function, a fundamental solution of an \emph{elliptic} partial differential operator. The Euclidean formalism not only simplifies calculations, but seems to be indispensable in constructive quantum field theory. The remarkable theorem of Osterwalder and Schrader gives sufficient conditions for the possibility to recover the original
hyperbolic (physically meaningful) field theory from a Euclidean framework, and therefore justifies the Euclidean framework in ordinary quantum field theory. It is recalled below how the Schwinger function of the Euclidean framework of free scalar field theory is derived by analytic continuation from 
the hyperbolic theory and how it relates to the Feynman propagator. 

To incorporate gravitational aspects into quantum field theory, one possibility is to study quantum fields on noncommutative spaces, the most popular of which is the noncommutative ``Moyal space'' whose coordinates are subject to commutation relations of the Heisenberg type~\cite{dfr}. Already in that early paper, a possible setting for (unitary) hyperbolic perturbative quantum field theory was proposed, where the field algebra is endowed with a noncommutative product, the twisted (convolution) product. Notwithstanding, the vast majority of publications on field theory on noncommutative spaces (``noncommutative field theory'') has been and still is formulated within a Euclidean setting. This setting was not derived from a hyperbolic noncommutative theory but directly from the Euclidean framework of ordinary  field theory by replacing all products with twisted ones. I shall refer to this approach as the traditional noncommutative Euclidean framework. Despite some attempts, it has not been possible to relate this traditional noncommutative Euclidean setting to some hyperbolic noncommutative theory -- in fact, there is evidence that it might be impossible to do so, unless the time variable commutes with all space variables. 
It became clear after some years that within the traditional Euclidean noncommutative framework, already the models built from the most harmless of fields, namely the scalar massive fields, have very peculiar properties.
%
%
Most notably, the so-called ultraviolet--infrared mixing problem noted in~\cite{uvir} severly limits the type of model that can be defined at all~\cite{GrW,R}.

In contrast to these results, I have shown~\cite{Bwien} that in a hyperbolic setting, the ultraviolet--infrared mixing effect is not present at least in the most prominent example graph that exibits ultraviolet--infrared mixing in the traditional Euclidean realm. This result will be presented in a longer and more technical article shortly. A general proof of the conjecture that the ultraviolet-infrared mixing problem may be absent in this hyperbolic noncommutative setting is, however, still missing,
as the calculations and the combinatorial aspects of hyperbolic noncommutative field theory are quite involved. It is therefore desirable to find a Euclidean framework that can actually be derived from a hyperbolic noncommutative setting in the hope that -- as in ordinary quantum field theory -- such a Euclidean setting might simplify  the combinatorial aspects of perturbation theory and that the full Euclidean machinery of renormalization might be employed. In such a setting, it might be feasable to investigate a theory's renormalizability and the possible absence of the ultraviolet-infrared mixing problem in general.

As a very first step towards this goal, I will show in this note that one can indeed derive a noncommutative Euclidean framework from a hyperbolic theory of free fields on the Moyal space, and that this framework is \emph{not} the traditional one that is investigated in the literature. In contrast to this traditional framework, the new Euclidean framework can moreover be related to a setting involving Feynman propagators via an analytic continuation similar to the one of ordinary quantum field theory. The note is organized as follows: In the next section it is recalled how the Schwinger function is derived in ordinary massive scalar quantum field theory and how it is related to the Feynman propagator. In the third section, a Euclidean 4-point function (Schwinger function) is derived from a noncommutative hyperbolic Wightman function of 4 free massive scalar fields and the prescription how Schwinger functions of arbitrarily high order are calculated is given. It is shown that the Euclidean framework thus derived differs from the traditional noncommutative Euclidean approach. Moreover, the relation to Feynman propagators is clarified. In an outlook I will briefly comment on further possible research that ensues from these new results.

\section{Euclidean methods in quantum field theory}

The hyperbolic partial differential operator of  massive scalar field theory is the massive Klein--Gordon operator $P:=\frac{\partial^2}{\partial {x_0}^2}-\Delta_\mathbf{x} +m^2$ on $\R^4$ where $\Delta_\mathbf{x}$ denotes the Laplace operator on $\R^3$, $\mathbf{x} \in \R^3$, and $m>0$ is a real parameter, called the field's mass. As mentioned in the introduction, all the relevant quantities of a scalar field theoretic model can be calculated from a fundamental solution of this operator. Recall here that 
a distribution $E \in  \mathcal D^\prime (\R^n)$ is a fundamental solution (or Green's function) of a partial linear differential operator $P(\partial)$ on $\R^n$ provided that in the sense of distributions, $P(\partial) E = - \delta$ with $\delta$ denoting the $\delta$-distribution.

Our starting point here, however, is the 2-point-function $\Delta_+\in \mathcal S^\prime (\R^4)$, a tempered distribution which is a solution (not a fundamental solution)
of the Klein--Gordon equation, $P\Delta_+=0$ in the sense of distributions. For $x=(x_0,\mathbf{x}) \in \R^4$, $x_0\in \R$, $\mathbf{x} \in \R^3$, it is given explicitly by 
\[
\Delta_+ (x) = \frac 1{(2\pi)^3} 
\int \,\mshf k \, \textup{e}^{-\textup{i} \omsh k x_0 + \textup{i} \mathbf{k}\mathbf{x}} \,  \mshi k \ , \qquad \mbox{where }\omsh k= \sqrt{\mathbf{k}^2+m^2},
\]
an expression which in fact makes sense as an oscillatory integral, see \cite[Sec~IX.10]{RS} for details. Here and in what follows, boldface letters denote elements of~$\R^3$ and an expression such as $\mathbf{k}\mathbf{x}$ is shorthand for the canonical scalar product of $\mathbf{k}$ and $\mathbf{x}$. 

It is well-known that $\Delta_+$ is the boundary value (in the sense of distributions) of an analytic function. To see this, let us first fix some notation.
Let $a \in \R^n$ with $|a|=1$, let $\theta \in (0,\pi/2)$, and let $ay$ denote the canonical scalar product in $\R^n$. Then the cone  about $a$  with opening angle~$\theta$ is the set $\Gamma_{a,\theta}=\{ y \in \R^n \mid ya>|y|\cos \theta \} \subset \R^n$. Let $\Gamma^*_{a,\theta}$ denote the dual cone, $\Gamma^*_{a,\theta}:=\Gamma_{a,\frac \pi 2-\theta}$.
For temperered distributions whose support is contained in the closure of a cone, the following general assertion holds:

\begin{thm}[\cite{RS}, Thm IX.16]
  \label{thm:AnaCont} Let $u$ be a tempered distribution with support in the closure of a cone ${\Gamma_{a,\theta}}$, $a \in \R^n$, $0<\theta < \pi/2$. Then its Fourier transform $\tilde u$ is the boundary value \textup{(}in the sense of tempered distributions\textup{)} 
of a function~$f$ which is analytic in the tube $\{z \in \C^n | -\im z \in \Gamma^*_{a,\theta} \}=:\R^n - \textup{i} \Gamma^*_{a,\theta} \subset \C^n$.
\end{thm}

\medskip
Observe that for $\tilde u$ to be the boundary value of $f$ as above in the sense of tempered distributions 
means that for any $\eta \in \Gamma^*_{a,\theta}$ and for any testfunction $g \in \mathcal S(\R^4)$, we have 
for $t \in \R$ approaching 0 from above, 
\[
\int f(x-\textup{i}t\eta)\,g(x)\,\textup{d}x \; \rightarrow \;  \tilde u (g) 
\]
as tempered distributions.

The Fourier transform $\tilde \Delta_+$ of the 2-point function, 
\begin{equation}
\tilde \Delta_+ (p_0,\mathbf{p})= \frac 1 {2 \omsh p}\,\delta(p_0-\omsh p), 
\end{equation}
is a tempered distribution whose support  (the positive mass shell) is contained in the closure of the cone $\Gamma_+:=\Gamma_{(1,\mathbf{0}), \pi/4}$ (the forward light cone).
Applied to $u:=\tilde \Delta_+$, Theorem~\ref{thm:AnaCont} thus guarantees that $\tilde u=\Delta_+$ is the boundary value of a function $f$ which is analytic in $\R^4 - \textup{i} \Gamma_+$ (observe that $\Gamma_+^*=\Gamma_+$). Explicitly, for $x=(x_0,\mathbf{x}) \in \R^4$ and $\eta=(x_4,\mathbf{0}) \in \Gamma_+$ (hence $x_4>0$), we have in this case
\begin{equation}\label{eq:schwingerT}
f(x-\textup{i}\eta) = \frac 1{(2\pi)^3} \int \mshf k \,
\textup{e}^{+\textup{i}\mathbf{k}\mathbf{x} - \omsh k (x_4+\textup{i}x_0)}\, \mshi k.
\end{equation}
We now define a function $s$ via 
\[
s(\mathbf{x}, x_4 +\textup{i}x_0):=f(x-\textup{i}\eta)
\]
for $x=(x_0,\mathbf{x}) \in \R^4$ and $\eta=(x_4,\mathbf{0}) \in \Gamma_+$ as above. Making use of the identity 
\begin{equation}\label{identity}
\frac 1{2 \omsh k}\; \textup{e}^{-\omsh k x_4} = \frac 1 {2\pi} \int_{-\infty}^\infty  \frac{\textup{e}^{\textup{i}k_4 x_4}}{k^2 +m^2} \; \textup{d}k_4 \qquad \mbox{ for } x_4 >0
\end{equation}
where $k=(\mathbf{k},k_4) \in \R^4$, $k^2=\mathbf{k}^2+k_4^2$, and setting $x_0=0$ in (\ref{eq:schwingerT}), we then find that 
\begin{equation}\label{eq:schwinger}
s(x) = \frac 1 {(2\pi)^4} 
\int  \frac{\textup{e}^{\textup{i}kx}}{k^2+m^2} \;\textup{d}^4k 
\quad \mbox{where $x=(\mathbf{x},x_4) \in \R^4$, $x_4>0$}.
\end{equation}
One now extends the function $s$ to a distribution $S \in \mathcal S^\prime (\R^4)$, the so-called \emph{Schwinger function}, by dropping the restriction on $x_4$. So, the formal integral kernel of $S$ 
is given by the Fourier transform\label{extension}
\[
\frac 1 {(2\pi)^4} 
\int  \frac{\textup{e}^{\textup{i}kx}}{k^2+m^2} \;\textup{d}^4k 
\]
of the smooth function
\[
\tilde S(k)=\frac 1 {k^2+m^2}
\]
on $\R^4$.
By definition, when restricted to the upper half space $x_4>0$, $S(\mathbf x, x_4)$ is (pointwise) equal to the function $s$ given in (\ref{eq:schwinger}).
Observe also that $S$ is the unique fundamental solution of the \emph{elliptic} partial differential operator $\Delta - m^2$ with $\Delta$ the Laplace operator on $\R^4$.

As mentioned in the introduction, the building block in hyperbolic perturbation theory is the Feynman propagator~$\Delta_F$, a fundamental solution for the Klein--Gordon operator $P=\frac{\partial^2}{\partial {x_0}^2}-\Delta_\mathbf{x} +m^2$.
Without going into details, let me mention that, remarkably, the Fourier transform $\tilde S$ of the Schwinger function is the analytic continuation of the Fourier transform $\tilde \Delta_F$ of the Feynman propagator $\Delta_F$ (up to a sign). In fact, formally, \label{SDeltaF} for the kernel $w$ given by 
\[
w(\mathbf{p},p_4-\textup{i}p_0):=\tilde \Delta_F(p_0+\textup{i}p_4,\mathbf{p}) \ ,
\]
we have 
\[
\tilde S(\mathbf{p},p_4)=-w(\mathbf{p},p_4) \phantom{\int_b}
\] 
for the Schwinger function's Fourier transform $\tilde S$.

\section{Analytic continuation in the noncommutative case}
It would be beyond the scope of this note to explain the possible unitary perturbative setups for massive scalar fields on the noncommutative Moyal space with hyperbolic signature (see~\cite{bahnsDiss} for a comparison). 
Only two features of such noncommutative (hyperbolic) field theories matter here. The first is the fact that our starting point still is the Klein--Gordon operator and the 2-point-function $\Delta_+$ discussed in the previous section. 
The  second important feature -- and this feature is shared by the traditional noncommutative Euclidean formalism -- is the fact that one has to consider not only products but also twisted products of distributions. 

To fix the notation, we note here that for two Schwartz functions $f,g  \in \mathcal S(\R^4)$ this twisted product (Moyal product) is
\begin{equation}\label{def:twist}
f * g \, (x) = 
\int  \, \tilde f (k) \; \tilde g (p)
\; \textup{e}^{-\textup{i}(p+ k)x} \; \textup{e}^{-\frac{\textup{i}}{2} p\theta k}\; \textup{d}^4k \, \textup{d}^4 p
\end{equation}
where $\tilde f$ and $\tilde g$ denote the Fourier transforms of $f$ and $g$, respectively, and where $\theta$ is a \emph{nondegenerate} antisymmetric $4\times 4$-matrix.
Observe that in a Euclidean theory, a product such as $kx$ stands for the canonical scalar product on $\R^4$, whereas in a hyperbolic setting, it denotes a Lorentz product, $kx =k_0x_0 - \mathbf{k}\mathbf{x}$ for $x=(x_0,\mathbf{x})$ and $k=(k_0,\mathbf{k})$, with $ \mathbf{k}\mathbf{x}$ denoting the canonical scalar product on $\R^3$. The oscillating factor $\textup{e}^{-\frac{\textup{i}}{2} p\theta k}$ is also called the twisting.

\subsection{The tensor product of 2-point functions}
Since the 2-point function remains unchanged in noncommutative field theory, we have to consider higher order correlation functions in order to see a difference between field theory on Moyal space and ordinary field theory. Again, it would be beyond the scope of this note to explain the whole setup. It will be sufficient to consider as an example a  particular contribution to the so-called 4\nobreakdash-point function of free massive scalar field theory. In ordinary field theory, the distribution of interest here is the  2-fold tensor product of 2\nobreakdash-point functions, 
\begin{equation}\label{def:Delta+2}
\Delta_+^{(2)} (x,y) = 
\frac 1{(2\pi)^6} \int \mshf k \, \mshf p \;
	 \textup{e}^{-\textup{i} (\omsh k x_0 + \omsh p y_0) 
	 		+\textup{i} (\mathbf{k} \mathbf{x} +  \mathbf{p} \mathbf{y})}\; \mshi k \, \mshi p
\end{equation}
The reader who is familiar with quantum field theory (in position space) will of course recognize that this tensor product makes up the 4\nobreakdash-point function (i.e. the vacuum expectation value of four fields), since
\[
\langle \Omega , \phi(x_1) \, \phi(x_2)\, \phi(x_3)\, \phi(x_4) \, \Omega \rangle \ = \
 \sum 
\Delta_+^{(2)} (x_{i_1}-x_{j_1},x_{i_2}-x_{j_2})
\]
where the sum runs over all pairs $(i_1,j_1)$, $(i_2,j_2)$ of indices with $\{i_1,i_2,j_1,j_2\}=\{1,2,3,4\}$ and $i_1<j_1$, $i_2<j_2$. 

By standard arguments from microlocal analysis involving the wavefront set of distributions, it can be shown that even the pullback of this tensor product with respect to the diagonal map, that is, the product in the sense of H\"ormander, is a well-defined distribution $\in \mathcal S^\prime (\R^4)$ (see for instance \cite[Chap~IX.10]{RS}). For the kernel given by (\ref{def:Delta+2}), this would amount to setting $x=y$. In order to avoid issues regarding renormalization later, in this note, however, only tensor products of distributions will be considered.

It is well-known and not difficult to see that $\Delta_+^{(2)}$ is again the boundary value of an analytic function:

\begin{lem} \label{lem:AnaCont2}  The tempered distribution $\Delta_+^{(2)}$ is the boundary value of a function $f_2$ which is analytic in  $\R^{4}\times \R^4-\textup{i}\Gamma_+\times \Gamma_+$. Explicitly, for
\[
z=(x_0,\mathbf{x},y_0,\mathbf{y}) \in \R^4\times \R^4
\qquad\text{and}\qquad
\eta=(x_4,\mathbf{0},y_4,\mathbf{0}) \in \Gamma_+\times \Gamma_+
\]
\textup{(}hence $x_4$ and $y_4>0$\textup{)}, we have 
\[ 
f_2(z-\textup{i}\eta) = \frac 1{(2\pi)^6} \int \mshf k \, \mshf p  
 \; \textup{e}^{- \omsh k (x_4+\textup{i}x_0) - \omsh p (y_4+\textup{i}y_0) +\textup{i} \mathbf{k} \mathbf{x}+\textup{i}\mathbf{p} \mathbf{y}}
\; \mshi k \, \mshi p \ .
\] 
For 
the function $s_2$ defined for $\eta$ and $z$ as above, by 
\[
s_2(\mathbf{x}, x_4 +\textup{i}x_0,\mathbf{y}, y_4 +\textup{i}y_0):=f_2(z-\textup{i}\eta)\ ,
\] 
we find for $(x,y)=(\mathbf{x}, x_4,\mathbf{y}, y_4) \in \R^4\times \R^4$, $x_4$ and $y_4>0$, the explicit form
\begin{equation}\label{eq:schwinger2}
s_2(x,y) = \frac 1{(2\pi)^8} \int 
\; \frac{1}{k^2+m^2} \;  \frac{1}{p^2+m^2}
 \; \textup{e}^{ +\textup{i} k x  +\textup{i} p y}
 \; \textup{d}^4k \, \textup{d}^4 p
\end{equation}
where $p^2=\mathbf{p}^2+p_4^2$, and likewise, $k^2=\mathbf{k}^2+k_4^2$.
\end{lem}

\begin{proof}The first claim is a direct consequence of Theorem~\ref{thm:AnaCont} applied with respect to $x$ and $y$ separately, and the second claim follows again from the identity (\ref{identity}).
\end{proof}

As in the previous section, one again dropps the restrictions on $x_4$ and $y_4$ and thereby extends $s_2$ to a distribution $S_2$, whose Fourier transform is the smooth function
\[
\tilde S(k)\,\tilde S(p)=\frac{1}{k^2+m^2}  \frac{1}{p^2+m^2}
\]
Again, upon restriction of $S_2$ to $\R^3\times \R_{>0} \times \R^3\times \R_{>0}$, it is equal to the function $s_2$. As an aside, it is mentioned that when one considers the pullback of $\Delta_+^{(2)}$ with respect to the diagonal map (such that, formally, one finds $x=y$ in (\ref{eq:schwinger2})),  then
the kernel $S_2(x,x)$ is the Fourier transform  of the convolution 
\[
\tilde S\times \tilde S\,(k)=\int  \frac{1}{(k-p)^2+m^2}\;  \frac{1}{p^2+m^2}
\; \textup{d}^4p\ . 
\]
Morevoer, let us consider again, how Feynman propagators enter the game. 
As is well-known, $\tilde S_2^\theta$ is the analytic continuation of a product of Feynman propagators. Explicitly, we find \label{S2DeltaF}
that its Fourier transform $\tilde S^\theta_2$ is given in terms of the kernel 
\[
w_2(\mathbf{k},k_4-\textup{i}k_0,\mathbf{p},p_4-\textup{i}p_0)
:=\tilde \Delta_F(k_0+\textup{i}k_4,\mathbf{k}) \;\tilde \Delta_F(p_0+\textup{i}p_4,\mathbf{p})
\]
as follows
\[
\tilde S_2(\mathbf{k},k_4,\mathbf{p},p_4)=-w_2(\mathbf{k},k_4,\mathbf{p},p_4)\ .
\]

It is well-known that the procedure applied to the twofold tensor product in lemma~\ref{lem:AnaCont2} can be applied more generally. Each contribution to the (hyperbolic) $2n$-point function (or Wightman function) is an $n$-fold tensor product of 2-point functions ($n$-point functions for odd $n$ vanish). In order to find the corresponding higher order Schwinger function, one considers the analytic continuation according to Theorem~\ref{thm:AnaCont} in each of the $n$ variables and proceeds in the same manner as explained for the 4-point function above.

\subsection{The twisted product of 2-point functions}
In~\cite{birkh}, it was shown how $2n$-point functions are calculated in hyperbolic massive scalar field theory on the noncommutative Moyal space ($n$-point functions for $n$ odd still vanish). As it turns out, the first deviation from ordinary field theory shows up in  the 4-point function, where one of the contributions is a {\em twisted} tensor product of two 2-point functions,
\begin{equation}\label{DeltaPlusNC}
\Delta_+^{(\star 2)}  (x,y) :=  \int \mshf k \, \mshf p 
 \; \textup{e}^{-\textup{i} (\omsh k x_0 + \omsh p y_0) +\textup{i} (\mathbf{k} \mathbf x +  \mathbf{p} \mathbf{x})}
 \; \textup{e}^{-\textup{i} \tilde p\theta \tilde k}
\; \mshi k \, \mshi p
\end{equation}
where $\tilde k=(\omsh k, \mathbf{k})$, and $\tilde p=(\omsh p, \mathbf{p})$.
In the terminology of physics, this means that the momenta $k$ and $p$ in the oscillating factor are \emph{on-shell}. This will turn out to be very important later on. It is also important to note that, while our starting point is the twisted product (\ref{def:twist}) on $\R^4$, the vectors in the twisting are on-shell 
as a consequence of the support properties of 
\[
\tilde \Delta_+(k_0,\mathbf{k})=\frac1{\omsh k}\; \delta(k_0-\omsh k) \ .
\] 
Observe also that  compared to the ordinary twisting in (\ref{def:twist}), the factor $2$ in the oscillating factor appears, since in the calculations, two oscillating factors as in (\ref{def:twist}) either cancel or (in the above case) add up, see~\cite{birkh}. 

Once more, we now apply Theorem~\ref{thm:AnaCont}.

\begin{lem} \label{lem:AnaContNC} The tempered distribution $\Delta_+^{(\star 2)}$ is the boundary value of a function $f^\theta_2$ which is analytic in $\R^{4}\times \R^4-\textup{i}\Gamma_+\times \Gamma_+$. Explicitly, for $z=(x_0,\mathbf{x},y_0,\mathbf{y}) \in \R^4\times \R^4$ and $\eta=(x_4,\mathbf{0},y_4,\mathbf{0}) \in \Gamma_+\times \Gamma_+$ \textup{(}hence $x_4$ and $y_4>0$\textup{)}, we have  
\[ 
f^\theta_2(z-\textup{i}\eta) = \frac 1{(2\pi)^6} \int \mshf k \, \mshf p  
 \; \textup{e}^{- \omsh k (x_4+\textup{i}x_0) - \omsh p (y_4+\textup{i}y_0) +\textup{i} \mathbf{k}\mathbf{x} +\textup{i}\mathbf{p} \mathbf{y}}
 \; \textup{e}^{-\textup{i} \tilde p\theta \tilde k} \; \mshi k \, \mshi p
\]
where $\tilde k=(\omsh k, \mathbf{k})$, $\tilde p=(\omsh p, \mathbf{p})$.
For 
the function $s^\theta_2$ defined for $\eta$ and~$z$ as above, by
\[
s^\theta_2(\mathbf{x}, x_4 +\textup{i}x_0,\mathbf{y}, y_4 +\textup{i}y_0)
:= f^\theta_2(z-\textup{i}\eta)
\]
we then find 
for $(x,y)=(\mathbf{x}, x_4,\mathbf{y}, y_4) \in \R^4\times \R^4$ with $x_4, y_4>0$,
\begin{equation}\label{eq:schwingerTheta2}
s^\theta_2(x,y) = \frac 1{(2\pi)^8} \int 
\; \frac{1}{k^2+m^2} \; \frac{1}{p^2+m^2}
 \; \textup{e}^{ +\textup{i} k x} \; \textup{e}^{ +\textup{i} p y}
 \; \textup{e}^{-\textup{i} \tilde p \theta  \tilde k}
\; \textup{d}^4k \, \textup{d}^4 p\end{equation}
where $p^2=\mathbf{p}^2+p_4^2$ and $p^2=\mathbf{p}^2+p_4^2$, and with $\tilde k=(\omsh k, \mathbf{k})$, $\tilde p=(\omsh p, \mathbf{p})$ as above.
\end{lem}

\begin{proof}Since the Fourier transform of $\Delta_+^{(\star 2)}$ is still a tempered distribution with support  contained in the closure of $\Gamma_+\times \Gamma_+$,
the first claim follows from Theorem~\ref{thm:AnaCont}. The second claim again follows from the identity (\ref{identity}) -- which, as should be noted, does not affect the twisting factor.
\end{proof}

Observe that $s^\theta_2$ and $s_2$ from Lemma~\ref{lem:AnaCont2} differ only by the oscillating factor 
$ \textup{e}^{-\textup{i} \tilde p\theta \tilde k}$. 
As before, we now extend $s^\theta_2$ to a distribution $S^\theta_2$ by dropping the restriction on $x_4$ and $y_4$, such that $S^\theta_2$ is given by the Fourier transform of the smooth function 
\begin{equation}\label{sTildeTheta2}
\tilde S_2^{\theta}(k,p)=\frac{1}{k^2+m^2} \; \frac{1}{p^2+m^2}
 \; \textup{e}^{-\textup{i} \tilde p\theta \tilde k}.
\end{equation}
Again, in the case of coinciding points, instead of $\tilde S_2^{\theta}(k,p)$ one considers the Fourier transform of the (now twisted) convolution 
\[
\int  \frac{1}{(k-p)^2+m^2} \; \frac{1}{p^2+m^2}
 \; \textup{e}^{-\textup{i} \tilde p\;\theta \; \widetilde{k-p}} 
 \; \textup{d}^4p
\]
where $\widetilde{k-p}=(\omsh{k-p},k-p)$.

\medskip
It is very important to note that the momenta which appear in the oscillating factors in all the expressions above are on-shell, i.e. that they are of the form   $\tilde p=(\omsh p, \mathbf{p})$, likewise for $k$ or $p-k$. 
The oscillating factor therefore distinguishes the components of $(\mathbf{p}, p_4)$ and is, in particular, independent of the fourth component~$p_4$. The reason for this lies in the fact that the Fourier transform of the 2-point function forces the momenta in the oscillating factor to be on-shell, and this is not changed by the analytic continuation.

\medskip These considerations turn out to be crucial in the following assertion:

\begin{rem} \label{sFeynman} Since the oscillating factor in (\ref{sTildeTheta2}) is independent of one of the components of $k$ and $p \in \R^4$, it is obvious that $\tilde S_2^\theta$ is the analytic continuation a product of Feynman propagators with an on-shell twisting. Explicitly, we find 
that the Schwinger function's $S^\theta_2$ Fourier transform $\tilde S^\theta_2$ is given in terms of the kernel 
\[
w^\theta_2(\mathbf{k},k_4-\textup{i}k_0,\mathbf{p},p_4-\textup{i}p_0)
:=\tilde \Delta_F(k_0+\textup{i}k_4,\mathbf{k}) \;\tilde \Delta_F(p_0+\textup{i}p_4,\mathbf{p})
\; 	\textup{e}^{-\textup{i} \tilde p \theta  \tilde k}
\]
as follows
\[
\tilde S^\theta_2(\mathbf{k},k_4,\mathbf{p},p_4)=-w^\theta_2 (\mathbf{k},k_4,\mathbf{p},p_4)\ .
\]
\end{rem}

\begin{rem} All this remains true when one calculates the higher order Schwinger functions from the $2n$-point functions. These latter distributions are of a similar form as (\ref{DeltaPlusNC}), i.e. they are twisted tensor products of 2-point functions where a certain combinatorics determines which combinations of momenta appear in the twistings, see~\cite{birkh}. The important point is that again, all momenta in the twistings are on-shell. Therefore, the same construction that was employed for the 4-point function above, i.e. an analytic continuation in the $n$ variables separately, can be applied and again leads to Schwinger functions with twistings that remain {\em on-shell}. Finally, the analytic continuation of the corresponding Fourier transform can be performed as in Remark~\ref{sFeynman} and leads to (twisted products of) Feynman propagators with twistings still only involving mass-shell momenta. 
\end{rem}

\medskip 
The fact that one starts from hyperbolic two-point functions which in turn force the momenta in the twistings to be on-shell is the essential difference to the traditional noncommutative Euclidean framework employed in the literature. In this latter framework,  (Euclidean) Schwinger functions are the starting point, and of course, when twisted products appear, by (\ref{def:twist}) the oscillating factors depend on all four components of a momentum vector $k=(\mathbf k, k_4)$.  For instance, instead of finding $\tilde s_2^{\theta}$ as in (\ref{sTildeTheta2}), one starts from the following  expression  
\begin{equation}\label{eTildeTheta2}
\tilde e_2^{\theta}(k,p)=\frac{1}{k^2+m^2} \; \frac{1}{p^2+m^2}
 \; \textup{e}^{-\textup{i} p\theta  k} 
\end{equation}
where $k=(\mathbf{k}, k_4)$ and $p=(\mathbf{p}, p_4)$. So far, it was not possible to relate this framework to a hyperbolic one, the main difficulty being the dependence of the oscillating factor on $k_4$. Naively copying the procedure sketched on page~\pageref{SDeltaF}, on page~\pageref{S2DeltaF} and in Remark~\ref{sFeynman} to pass to Feynman propagators (via the kernels $w$, $w_2$ and $w_2^\theta$, repsectively) 
leads to exponentially increasing terms which render the integrals ill-defined. So far, the only way out found seems to be to make the oscillating factor independent of one of the components in an \emph{ad hoc} way, by requiring~$\theta$ to be a matrix of rank~2 (``spacelike noncommutativity'').

\medskip

Remark~\ref{sFeynman} shows that such  measures are unnecessary  when the new noncommutative Euclidean framework derived from the hyperbolic $n$-point functions is employed.

\section{Outlook}

Based on the above considerations, the most important question now is whether it is possible to set up a consistent Euclidean noncommutative framework with on-shell momenta in the twisting. An obstruction might be that, 
as can be easily seen already in the example $S_2^\theta$  discussed above, the higher order Schwinger functions are not symmetric with respect to reflections in the origin. Also, the new Euclidean on-shell product is not associative. This may jeopardize the possibility to set up a complete consistent perturbative framework using a Schwinger functional.

Still, it is to be hoped that the results presented here open many interesting possibilities for future research. 
For one thing, one should try to generalize the Osterwader Schrader Theorem in this setting. Also,  it would be most interesting to study whether the ultraviolet-infrared mixing problem appears in this setting at all. Certainly,there is reason to hope so, since the most prominent graph (the nonplanar tadpole) that exibits this problem in the traditional Euclidean noncommutative approach, does {\em not} do so, when one simply replaces its twisting by an on-shell twisting.

Last but not least, a thorough understanding of the new Euclidean setup (if feasable) should enable us to learn more about hyperbolic noncommutative models -- which in themselves have proved to be quite difficult to treat.  If a consistent Euclidean perturbative setup can be developed from  the ideas presented here,  general proofs of renormalizability of hyperbolic noncommutative field theory should at last be possible.

\backmatter
\bibliographystyle{smfplain.bst}
\bibliography{bahns_eu}

\end{document}